\documentclass[11pt]{article}
\usepackage{fullpage,amsmath,amsthm,amssymb,bbm}

\newtheorem{theorem}{Theorem}
\newtheorem{lemma}[theorem]{Lemma}

\newtheorem{proposition}[theorem]{Proposition}
\newtheorem{definition}[theorem]{Definition}
\newtheorem{question}[theorem]{Question}
\newtheorem{corollary}[theorem]{Corollary}

\newcommand{\norm}[1]{\left|\left| #1 \right|\right|}

\newcommand{\R}{{\mathbb{R}}}

\newcommand{\C}{{\mathcal{C}}}

\newcommand{\eps}{\varepsilon}

\newcommand{\poly}{\operatorname{poly}}

\newcommand{\ex}[2]{\underset{#1}{\mathbb{E}}\left[ #2 \right]}
\newcommand{\pr}[2]{\underset{#1}{\mathbb{P}}\left[ #2 \right]}

\newcommand{\sgn}{\operatorname{sgn}}
\newcommand{\LOL}{LSL}
\newcommand{\omitted}[1]{}

\title{Weighted Polynomial Approximations:\\Limits for Learning and Pseudorandomness}
\author{Mark Bun\thanks{Harvard University, School of Engineering and Applied Sciences. Supported by an NDSEG Fellowship and NSF grant CNS-1237235.} \and Thomas Steinke\thanks{Harvard University, School of Engineering and Applied Sciences. Supported
by NSF grant CCF-1116616 and the Lord Rutherford Memorial Research Fellowship.}}
\date{\texttt{\{mbun,tsteinke\}@seas.harvard.edu}\\~\\December 8, 2014}

\begin{document}
\maketitle

\begin{abstract}
Polynomial approximations to boolean functions have led to many positive results in computer science. In particular, polynomial approximations to the sign function underly algorithms for agnostically learning halfspaces, as well as pseudorandom generators for halfspaces. In this work, we investigate the limits of these techniques by proving inapproximability results for the sign function.

Firstly, the ``polynomial regression'' algorithm of Kalai et al. (SIAM J. Comput. 2008) shows that halfspaces can be learned with respect to log-concave distributions on $\R^n$ in the challenging agnostic learning model. The power of this algorithm relies on the fact that under log-concave distributions, halfspaces can be approximated arbitrarily
well by low-degree polynomials. We ask whether this technique can be extended beyond log-concave distributions, and establish a negative result. We show that polynomials of any degree cannot approximate the sign function to within arbitrarily low error for a large class of non-log-concave distributions on the real line, including those with densities proportional to $\exp(-|x|^{0.99})$. This impossibility result extends to multivariate distributions, and thus gives a strong limitation on the power of the polynomial regression algorithm for halfspaces.

Secondly, we investigate the derandomization of Chernoff-type concentration inequalities. Chernoff-type tail bounds on sums of independent random variables have pervasive applications in theoretical computer science. Schmidt et al. (SIAM J. Discrete Math.~1995) showed that these inequalities can be established for sums of random variables with only $O(\log(1/\delta))$-wise independence, for a tail probability of $\delta$. We show that their results are tight up to constant factors.

These results rely on techniques from weighted approximation theory, which studies how well functions on the real line can be approximated by polynomials under various distributions. We believe that these techniques will have further applications in other areas of theoretical computer science.
\end{abstract}

\section{Introduction}
Approximation theory is a classical area of mathematics that studies how well functions can be approximated by simpler ones. It has found many applications in computer science. Most of these applications of approximation theory focus on the approximation of functions by polynomials in the uniform norm (or infinity norm). For instance, \emph{approximate degree}, which captures how well a boolean function can be approximated by low-degree polynomials in the uniform norm, underlies important lower bounds in circuit complexity \cite{Beigel93, Beigel94, Sherstov09}, quantum query complexity \cite{BealsBuClMoWo01, AaronsonSh04}, and communication complexity \cite{Sherstov08}. It also underlies state-of-the art algorithms in learning theory \cite{KalaiKlMaSe08, KlivansSe04}, streaming \cite{HarveyNeOn08}, and in spectral methods \cite{SachdevaVi14}.

While it is compelling to study polynomial approximations under the uniform norm, there are scenarios where it is more natural to study \emph{weighted polynomial approximations}, where error is measured in terms of an $L_p$ norm under some distribution. For instance, in agnostic learning, the polynomial regression algorithm of Kalai et al.~\cite{KalaiKlMaSe08} has guarantees based on how well functions in a concept class of interest can be approximated by low-degree polynomials in $L_1$ distance. 

In this work, we show how ideas from weighted approximation theory can yield tight lower bounds for several problems in theoretical computer science. As our first application, we establish a strong limitation on the distributions under which halfspaces can be learned using the polynomial regression algorithm of Kalai et al. Second, in the area of derandomization, we give a tight characterization of the amount of $k$-wise independence necessary to establish Chernoff-like concentration inequalities.

\subsection{Agnostically Learning Halfspaces}

Halfspaces are a fundamental concept class in machine learning, both in theory and in practice.\footnote{A halfspace is a function $f : \mathbb{R}^n \to \{\pm 1\}$ given by $f(x) = \sgn(w \cdot x - \theta)$ for $w \in \mathbb{R}^n$ and $\theta \in \mathbb{R}$, where $\sgn(x) = 1$ if $x \geq 0$ and $\sgn(x)=-1$ otherwise.} Their study dates back to the Perceptron algorithm of the 1950s. Halfspaces serve as building blocks in many applications, including boosting and kernel methods.

Halfspaces can be learned in the PAC model \cite{Valiant84} either by solving a linear program, or via simple iterative update algorithms (e.g. the Perceptron algorithm). However, learning halfspaces with classification noise is a much more difficult problem, and often needs to be dealt with in practice.

In this work, we study a challenging model of \emph{adversarial noise} -- the agnostic learning model of Kearns et al. \cite{KearnsScSeHe94}. In this model, a learner has access to examples drawn from a distribution $\mathcal{D}$ on $X \times \{\pm 1\}$ and must output a hypothesis $h : X \to \{\pm 1\}$ such that $$\pr{(x,y) \sim \mathcal{D}}{h(x) \ne y} \leq \mathrm{opt} + \varepsilon,$$ where $\mathrm{opt}$ is the error of the best concept in the concept class -- that is, $\mathrm{opt} = \min_{f \in \mathcal{C}} \pr{(x,y) \sim \mathcal{D}}{f(x) \ne y}$.

The theory of agnostic learning is not well-understood, even in the case of halfspaces. Positive results for efficient agnostic learning of high-dimensional halfspaces are restricted to limited classes of distributions.\footnote{An efficient algorithm is one which runs in time polynomial in the dimension $n$ for any constant $\varepsilon>0$ -- that is, time $n^{O_\varepsilon(1)}$.}  For instance, halfspaces can be learned under the uniform distribution over the hypercube or the unit sphere, or on any log-concave distribution \cite{KalaiKlMaSe08}. On the negative side, a variety of both computational and information-theoretic hardness results are known. For instance, proper agnostic learning of halfspaces (where the learner is required to output a hypothesis that is itself a halfspace) is known to be NP-hard \cite{FeldmanGoKhPo06}. Moreover, agnostically learning halfspaces under arbitrary distributions is as hard as PAC learning DNFs \cite{LeeBaWi95}, which is a longstanding open problem.

There is essentially only one known technique for agnostically learning high-dimensional halfspaces: the $L_1$ regression algorithm \cite{KalaiKlMaSe08}, which we discuss in more detail in Section \ref{sec:regression}. In its most general form, the algorithm selects a linear space of functions $\mathcal{H} \subset \{ h : X \to \mathbb{R} \}$. After drawing a number of examples $(x_i,y_i)$ from $\mathcal{D}$, it computes $$ h^* = \underset{h \in \mathcal{H}}{\mathrm{argmin}} \sum_i |h(x_i) - y_i|.$$ The output of the algorithm is $\sgn(h^*(x)-t)$ for some $t$. We need to ensure that the minimisation can be computed efficiently (e.g. by linear programming) and that every concept $f \in \mathcal{C}$ can be approximated by some $h \in \mathcal{H}$ -- that is $ \ex{x \sim \mathcal{D}}{|h(x)-f(x)|} \leq \varepsilon$. If this is the case, then $\C$ is agnostically learnable in time $\poly(|\mathcal{H}|)$.

Kalai et al. (and most subsequent work on learning using $L_1$ regression, e.g. \cite{KlivansODSe08, GopalanKaKl08, BlaisODWi10, KaneKlMe13, FeldmanKo14}) chose $\mathcal{H}$ to be the class of low-degree polynomials. They showed that under certain classes of distributions, every halfspace can be approximated by a polynomial of degree $O_\eps(1)$, and hence halfspaces are agnostically learnable in time $n^{O_\eps(1)}$.

Distributional assumptions arise because we use an $L_1$ approximation measure (namely $ \ex{x \sim \mathcal{D}}{|h(x)-f(x)|} $), which depends on the distribution. A distribution-independent approximation would require an $L_\infty$ approximation, which is too much to hope for in many circumstances.


\subsubsection{Our Results}

Can we weaken the distributional assumptions required for learning halfspaces using current techniques?
Our result addressing this question (Theorem \ref{thm:informal-noapprox}) is a negative one. We show that polynomial approximations to halfspaces do not exist for a large class of distributions, namely:

\begin{definition}
An absolutely continuous distribution $\mathcal{D}$ on $\mathbb{R}$ is a \emph{log-superlinear (\LOL) distribution} if there exist $C > 0$ and $\gamma \in (0,1)$ such that the density $w$ of $\mathcal{D}$ satisfies $w(x) \ge C \exp(-|x|^\gamma)$.\footnote{The name log-superlinear comes from the fact that the tails of the probability density function of a \LOL\ distribution are heavier than that of the log-linear Laplace distribution.} 
\end{definition}

\begin{theorem}\label{thm:informal-noapprox}
For any \LOL\ distribution $\mathcal{D}$, there exists $\varepsilon>0$ such that no polynomial (of any degree) can approximate the sign function with $L_1$ error less than $\varepsilon$ with respect to $\mathcal{D}$.
\end{theorem}

In particular, this implies that the polynomial regression algorithm is not able to agnostically learn thresholds on the real line to within arbitrarily small error. Note that this result does not rule out the possibility that halfspaces can be agnostically learned by other techniques. Indeed, the classic approach of empirical risk minimization (see \cite{KearnsScSeHe94} and the references therein) gives an efficient algorithm for learning thresholds (which are halfspaces in one dimension) under arbitrary distributions. Thus the problem of learning real thresholds under \LOL\ distributions is an explicit example for which polynomial regression fails while other techniques can succeed.

If we were to take $\gamma \geq 1$, the probability density function $C(\gamma) e^{-|x|^\gamma}$ (where $C(\gamma)$ is a normalising constant) would give a log-concave distribution, in which case Kalai et al.~\cite{KalaiKlMaSe08} show that good polynomial approximations to halfspaces exist. Thus our result gives a threshold between where polynomial approximations to halfspaces exist and where they do not.

Our result for thresholds extends readily to an impossibility result for learning halfspaces over $\R^n$:

\begin{theorem}
For any product distribution $\mathcal{D}$ on $\mathbb{R}^n$ with a \LOL\ marginal distribution on some coordinate, there exists $\varepsilon>0$ and a halfspace $h$ such that no polynomial can approximate $h$ with $L_1$ error less than $\varepsilon$ with respect to $\mathcal{D}$.
\end{theorem}

Our result echoes prior work establishing the limits of \emph{uniform} polynomial approximations for various concept classes. For instance, the seminal work of Minsky and Papert \cite{MinskyPapert} showed that there is an \emph{intersection} of two halfpsaces over $\mathbb{R}^n$ which cannot  be represented as the sign of any polynomial. Building on work of Nisan and Szegedy \cite{NisanSz94}, Paturi \cite{Paturi92} gave tight lower bounds for uniform approximations to symmetric boolean functions. This, and subsequent work on lower bounds for approximate degree, immediately imply limitations for distribution-independent agnostic learning via polynomial regression. Klivans and Sherstov \cite{KlivansSh10} also showed a strong generalization of Paturi's result to disjunctions, giving limitations on how well they can be approximated by linear combinations of arbitrary features. By contrast to all of these results, our work shows a strong limitation for certain \emph{distribution-dependent} polynomial approximations.

In the distribution-dependent setting, Feldman and Kothari \cite{FeldmanKo14} showed that polynomial regression cannot be used to learn disjunctions with respect to symmetric distributions on the hypercube. Recent work of Daniely et al. \cite{DanielyLiSh14} also uses ideas from approximation theory to show limitations on broad class of regression and kernel-based methods for learning halfspaces, even under a margin assumption. While our results only apply to polynomial regression, they hold for approximations of arbitrarily high complexity (i.e. degree), and for a large class of natural distributions.

The proof of Theorem \ref{thm:informal-noapprox} relies on several Markov-type inequalities for weighted polynomial approximations. These are generalizations of the classical Markov inequality for uniform approximations, which gives a bound on the derivative of a low-degree polynomial that is bounded on the unit interval:
\begin{theorem}[\cite{Markov90}]
Let $p$ be a polynomial of degree $d$ with $|p(x)| \le 1$ on the interval $[-1, 1]$. Then $|p'(x)| \le d^2$ on $[-1, 1]$.
\end{theorem}
Early work on the approximate degree of boolean functions \cite{NisanSz94, Paturi92} used Markov's inequality to get tight lower bounds on the degree of uniform approximations to symmetric functions. For weighted approximations under \LOL\ distributions, we actually get a much stronger statement. It turns out that under \LOL\ distributions, the derivative of a bounded polynomial near the origin is at most a constant \emph{independent of degree}. With this powerful fact in hand, the proof of Theorem \ref{thm:informal-noapprox} is quite simple. Consider the threshold function $f(t) = \sgn(t)$. Since $f$ has a ``jump'' at zero, any good polynomial approximation to $f$ must be bounded and have a large derivative near zero. The higher quality the approximation, the larger a derivative we need. But since the derivative of any polynomial is bounded by a constant, we cannot get arbitrarily good approximations to $f$ using polynomials.

We give the full proof in Section \ref{sec:lowerbound}, and discuss the multivariate generalization in Section \ref{sec:multi-lowerbound}.

\subsubsection{Related Work}




There is a rich literature on lower bounds for agnostic learning. In the case of \emph{proper} agnostic learning Feldman et. al \cite{FeldmanGoKhPo06} gave an optimal NP-hardness result for even weakly agnostically learn halfspaces over $\mathbb{Q}^n$. Guruswami and Raghavendra \cite{GuruswamiRa06} showed that the same is true even for halfspaces on the boolean hypercube. 

There has also been a line of work giving representation-independent hardness of learning halfspaces based on cryptographic assumptions. Feldman et. al \cite{FeldmanGoKhPo06} and Klivans and Sherstov \cite{KlivansSh09} showed that, assuming the security of certain public key encryptions schemes, it is hard to even PAC learn thresholds and intersections of halfspaces, respectively. These results imply that it is hard to agnostically learn a single halfspace in the harsh noise regime, i.e. when $\mathsf{opt}$ is very close to $\frac{1}{2}$. Shalev-Schwartz et al. \cite{ShalevShSr11} further showed that halfspaces cannot be efficiently learned even under a large margin assumptions.

There has also been extensive work proving unconditional lower bounds for restricted learning algorithms. One well-studied restriction on a learner is that it operates in Kearns' statistical query (SQ) model \cite{Kearns98, BlumFuJaKeMaRu94, KlivansSh07, FeldmanLeSe11}. This model captures $L_1$ regression, as well as essentially every technique known for learning (besides Gaussian elimination). Very recently, Dachman-Soled et al. \cite{DachmanFeTaWaWi14} showed that polynomial regression is in fact essentially the optimal SQ algorithm for agnostic learning with respect to product distributions on the hypercube.

The limitations we prove for polynomial regression do not rule out the existence of other agnostic learning algorithms, including those using $L_1$ regression with different feature spaces. Wimmer \cite{Wimmer10} showed how to use a different family of basis functions to learn halfspaces over symmetric distributions on the hypercube. Subsequent work of Feldman and Kothari \cite{FeldmanKo14} improved the running time in the special case of disjunctions. We leave it as an intriguing open question to determine whether other basis functions can be used to learn halfspaces under \LOL\ distributions.


\subsection{Tail Bounds for Limited Independence}

The famous Hoeffding bound \cite{Hoeffding} implies that if $X \in \{\pm 1\}^n$ is a uniform random variable and $r \in \mathbb{R}^n$ is fixed, then, for all $T \geq 0$, $$\pr{X}{|X \cdot r| \geq T} \leq 2 e^{-\frac{T^2}{2\norm{r}_2^2}}.$$
We ask the following question: 
\begin{center}
For what \emph{pseudorandom} $X$ is the Hoeffding bound true?
\end{center}
More precisely, given $T$ and $\delta$, can we construct a pseudorandom $X$ such that $\pr{X}{|X \cdot r| \geq T} \leq \delta$ for all $r \in \{\pm 1\}^n$?\footnote{For simplicity we restrict our attention to $r \in \{\pm 1\}^n$.} Of particular interest is the parameter regime $\delta = 1/\poly(n)$ and $T = \Theta(\norm{r}_2\sqrt{\log(1/\delta)})$. The probabilistic method gives a non-constructive proof that there exists such an $X$ which can be sampled with seed length $O(\log(n/\delta))$. The challenge is to give an explicit construction of such an $X$ which can be \emph{efficiently} sampled with a short seed.

This is a very natural pseudorandomness question: Concentration of measure is a fundamental property of independent random variables and one of the key objectives of pseudorandomness research is to replicate such properties for variables with low entropy. Finding a pseudorandom $X$ exhibiting good concentration is also a relaxation of a more general and well-studied pseudorandomness question, namely constructing pseudorandom generators that fool linear threshold functions \cite{DGJSV09,MekaZuckerman,GOWZ,DSTW}. This can also be viewed as a special case of constructing pseudorandom generators for space-bounded computation \cite{Nisan,INW,Reingold,BRRY,BrodyVerbin,KNP,RSV13}.

For $\delta = 1/\poly(n)$ and $T = \Theta(\sqrt{n \log(1/\delta)})$, we can construct generators $X$ with seed length $O(\log^2 n)$ using a variety of methods (including \cite{Nisan,MekaZuckerman}). In particular, it suffices for $X$ to be $O(\log(1/\delta))$-wise independent:

\begin{theorem}[Tail Bound for Limited Independence] \label{thm:KwiseUpperBound}
Let $n \geq 1$, $\eta > 0$, and $\delta \in (0,1)$ be given. Let $X \in \{\pm 1\}^n$ be $k$-wise independent for $k = 2 \lceil  \eta \log_e(1/\delta) \rceil$. Set $T = e^{(\eta+1)/2\eta}\sqrt{k} \norm{r}_2$. Then, for all $r \in \mathbb{R}^n$, $$\pr{}{|X \cdot r| \geq T} \leq \delta.$$
\end{theorem}

A $k$-wise independent $X \in \{\pm 1\}^n$ can be sampled with seed length $O(k \cdot \log n)$ \cite{ABI85}. Another construction which achieves seed length $O(\log n \cdot \log(1/\delta))$ is to sample $X$ from a small-bias space \cite{NaorN93}.  Very recently, Gopalan et al. \cite{GopalanKM14} constructed a new generator with seed length $\tilde{O}(\log(n/\delta))$, which is nearly optimal.

In this work, we ask whether the tail bound of Theorem \ref{thm:KwiseUpperBound} for $k$-wise independence is tight. That is, can we prove stronger tail bounds for $k$-wise independent $X$? 

\begin{question} \label{q:Kwise}
How much independence is needed for $X$ to satisfy a Hoeffding-like tail bound? That is, what is the minimum $k = k(n, \delta, T)$ for which any $k$-wise independent $X \in \{\pm 1\}^n$ satisfies
$$\pr{X}{|X \cdot r| \geq T} \leq \delta$$
for all $r \in \{-1, 1\}^n$, where $\cdot$ denotes the inner product.
\end{question}

\subsubsection{Our Results}

Theorem \ref{thm:KwiseUpperBound} shows that $k(n,\delta,T) \leq O(\log(1/\delta))$ for $T = O(\sqrt{n\log(1/\delta)})$. In this work, we show that this is essentially tight:

\begin{theorem} \label{thm:KwiseLowerBound}
For $T= c \sqrt{n \log(1/\delta)}$ ($c > 5$), we have $k(n,\delta,T) = \Omega(\log_c(1/\delta))$ for sufficiently large $n$.
\end{theorem}

The only previous lower bound was $$k(n,\delta,T)) \geq \Omega \left( \frac{\log(1/\delta)}{\log n} \right),$$ which holds for any $T \leq n$ and is due to \cite{SSS}. This is useful if $\delta < n^{-\omega(1)}$, but the lower bound is constant in our parameter regime. This lower bound follows from the fact that a random variable $X$ with support size $s$ cannot give a tail bound with $\delta<1/s$, and that there exist $k$-wise independent distributions with support size $s \leq O(n^k)$.

\medskip

The most natural way to prove Theorem \ref{thm:KwiseLowerBound} would be to construct a family of $k$-wise independent distributions that do not satisfy the required tail bound.  However, we instead study the \emph{dual} formulation of the problem (following \cite{Bazzi,DETT}) and then use lower bound techniques from approximation theory. To the best of our knowledge, this indirect approach is novel. Our results imply the existence of $k$-wise independent distributions with poor tail bounds, but give no immediate indication as to how to construct them!

We now describe the proof idea in slightly more detail. The answer to Question \ref{q:Kwise} can be posed in terms of the value of a certain linear program. The variables represent the probability distribution of the random variable $X$ and the constraints force $X$ to be $k$-wise independent. The objective of the linear program is maximize $\pr{}{|X \cdot r| \geq T}$. Thus, the value of the program is at most $\delta$ if and only if $k \geq k(n,\delta,T)$. Taking the dual of this linear program and appealing to strong duality yields an alternative characterization of $k(n, \delta, T)$. Namely, $k(n, \delta, T)$ is the smallest $k$ for which the threshold function $F_T(x) = \mathbbm{1}(|x| \ge T)$ admits an \emph{upper sandwiching polynomial} of degree $k$ and expectation at most $\delta$. Here, an upper sandwiching polynomial is simply a polynomial $p$ for which $p(x) \ge F_T(x)$ pointwise.

We then use ideas from weighted approximation theory to give a lower bound on $k$ for which such sandwiching polynomials exist. In order to apply these ideas, we make a few symmetrization and approximation arguments to reduce the problem to a continuous one-dimensional problem: Find a degree lower bound for a univariate polynomial that is a good upper sandwich for the function $f_T(x) = \sgn(|x| - T)$, with respect to a Gaussian distribution. As in our proof of Theorem \ref{thm:informal-noapprox}, the solution of this problem appeals to a weighted Markov-type inequality. Again, the idea is that an upper sandwich for $f_T$ must have a large jump at the threshold $T$, which is impossible for low-degree polynomials. The formal proof of this claim is based on a variant of an ``infinite-finite range'' inequality, which asserts that the weighted norm of a polynomial on the real line is bounded by its norm on a finite interval.

\section{Agnostically Learning Halfspaces}
The class of log-concave distributions over $\R^n$ (defined below) is essentially the broadest under which we know how to agnostically learn halfspaces. While many distributions used in machine learning are log-concave, such as the normal, Laplace, beta, and Dirichlet distributions, log-concave distributions do not capture everything. For instance, the log-normal distribution and heavier-tailed exponential power law distributions are not log-concave. The main motivating question for this section is whether we can relax the assumption of log-concavity for agnostically learning halfspaces. To this end, we show a negative result: for \LOL\ distributions, agnostic learning of halfspaces will require new techniques.

\subsection{Background}

Our starting point is the work of Kalai et al. \cite{KalaiKlMaSe08}. Among their results is the following.

\begin{theorem}[\cite{KalaiKlMaSe08}] \label{thm:aglearnhalf}
The concept class of halfspaces over $\R^n$ is agnostically learnable in time $\poly(n^{O_\varepsilon(1)})$ under log-concave distributions.
\end{theorem}

A log-concave distribution is an absolutely continuous probability distribution such that the logarithm of the probability density function is concave. For example, the standard multivariate Gaussian distribution on $\mathbb{R}^n$ has the probability density function $x \mapsto e^{-||x||_2^2/2}/(2 \pi)^{n/2}$. The natural logarithm of this is $-||x||_2^2/2 - n/2 \cdot \log (2\pi)$, which is concave. The class of log-concave distributions also includes the Laplace distribution and other natural distributions. However, it does not contain heavy-tailed distributions (such as power laws) nor non-smooth distributions (such as discrete probability distributions).

Kalai et al.~also show that we can agnostically learn halfspaces under the uniform distribution over the hypercube $\{\pm 1\}^n$ or over the unit sphere $\{x \in \mathbb{R}^n : ||x||_2=1\}$.

\subsection{The $L_1$ Regression Algorithm} \label{sec:regression}

The results of Kalai et al. are based on the so-called $L_1$ regression algorithm, which relies on being able to approximate the concept class in question by a low-degree polynomial:

\begin{theorem}[{\cite{KalaiKlMaSe08}}] \label{thm:L1}
Fix a distribution $\mathcal{D}$ on $X \times \{\pm 1\}$ and a concept class $\mathcal{C} \subset \{ f : X \to \{\pm 1\} \}$.\footnote{Here $X=\mathbb{R}^n$.} Suppose that, for all $f \in \mathcal{C}$, there exists a polynomial $p : X \to \mathbb{R}$ of degree at most $d$ such that $\ex{x \sim \mathcal{D}_X}{|p(x)-f(x)|} \leq \varepsilon$, where $\mathcal{D}_X$ is the marginal distribution of $\mathcal{D}$ on $X$. Then, with probability $1-\delta$ the $L_1$ regression algorithm outputs a hypothesis $h$ such that $$\pr{(x,y) \sim \mathcal{D}}{h(x) \ne y} \leq \min_{f \in \mathcal{C}} \pr{(x,y) \sim \mathcal{D}}{f(x) \ne y} + \varepsilon$$ in time $\poly(n^d,1/\varepsilon, \log(1/\delta))$ with access only to examples drawn from $\mathcal{D}$.
\end{theorem}

The $L_1$ regression algorithm solves a linear program to find a polynomial $p$ of degree at most $d$ that minimises $\sum_i |p(x_i)-y_i|$, where $(x_i,y_i)$ are the examples sampled from $\mathcal{D}$. The hypothesis is then $h(x) = \sgn(p(x)-t)$, where $t \in [-1,1]$ is chosen to minimise the error of $h$ on the examples.

Given Theorem \ref{thm:L1}, proving Theorem \ref{thm:aglearnhalf} reduces to showing that halfspaces can be approximated by low-degree polynomials under the distributions we are interested in. It is important to note that making assumptions on the distribution is necessary (barring a major breakthrough): Agnostically learning halfspaces under arbitrary distributions is at least as hard as PAC learning DNF formulas \cite{LeeBaWi95}. Moreover, proper learning of halfspaces under arbitrary distributions is known to be NP-hard \cite{FeldmanGoKhPo06}.

In fact, we can reduce the task of approximating a halfspace to a one-dimensional problem. A halfspace is given by $f(x) = \sgn(w \cdot x - \theta)$ for some $w \in \mathbb{R}^n$ and $\theta \in \mathbb{R}$. It suffices to find a univariate polynomial $p$ of degree at most $d$ such that $\ex{x \sim \mathcal{D}_{w,\theta}}{|p(x)-\sgn(x)|} \leq \varepsilon$, where $\mathcal{D}_{w,\theta}$ is the distribution of $w \cdot x - \theta$ when $x$ is drawn from $\mathcal{D}_X$. If $\mathcal{D}_X$ is log-concave, then so is $\mathcal{D}_{w,\theta}$.


\subsection{On the Density of Polynomials}

In this section, we give some intuition for why one might expect that polynomial approximations do not suffice for learning under \LOL\ distributions. It turns out that under a \LOL\ distribution $w$, polynomials actually fail to be dense in the space $C_0[w]$ of continuous functions vanishing at infinity when weighted by $w$. This is in stark contrast to the classical Weierstrass approximation theorem, which asserts that the polynomials are dense in $C_0$ under the uniform weight. These kinds of results address \emph{Bernstein's approximation problem} \cite{Bernstein24}, a precise statement of which is as follows.

\begin{question}
Let $w: \R \to [0, 1]$ be a measurable function. Let $C_0[w]$ denote the space of continuous functions $f$ for which $\lim_{|x| \to \infty} f(x)w(x) = 0$. Under what conditions on $w$ is it true that for every $f \in C_0[w]$, there is a sequence of polynomials $\{p_n\}_{n=1}^\infty$ for which
\[\lim_{n \to \infty}\|(p_n - f)w\|_\infty = 0?\]
\end{question}
(The choice of the $L_\infty$ norm here appears to make very little difference). If Bernstein's problem admits a positive resolution, we say that the polynomials are \emph{dense} in $C_0[w]$. The excellent survey of Lubinsky \cite{Lubinsky07} presents a number of criteria for when polynomials are dense. The one that is most readily applied was proved by Carleson \cite{Carleson51} (but appears to be implicit in \cite{IzumiKa37}):

\begin{theorem}
Let $w$ be even and positive with $\log(w(e^x))$ concave. Then the polynomials are dense in $C_0[w]$ iff
\[\int_0^\infty \frac{\log w(x)}{1 + x^2} \mathrm{d}x = -\infty.\]
\end{theorem}

This immediately yields the following dichotomy result for exponential power distributions:

\begin{corollary}
For $\gamma > 0$ and $w_\gamma(x) = \exp(-|x|^\gamma)$, the polynomials are dense in $C_0[w_\gamma]$ iff $\gamma \ge 1$.
\end{corollary}

In particular, this justifies our assertion that the polynomials fail to be dense in the continuous functions under \LOL\ distributions.

So what does this have to do with agnostically learning halfspaces? Recall that the analysis of the $L_1$-regression algorithm of Kalai et al. \cite{KalaiKlMaSe08} reduces approximating a halfspace under a distribution $\mathcal{D}$ to the problem of approximating each threshold function $\sgn(x - \theta)$ under each marginal distribution of $\mathcal{D}$. So for the algorithm to work, we require $\mathcal{D}$ to have marginals $w$ under which $\sgn(x - \theta)$ can be approximated arbitrarily well by polynomials. Now if the polynomials are dense in $C_0[w]$, then threshold functions can also be approximated arbitrarily well (since $C_0[w]$ is in turn dense in $L_1[w]$). Such an appeal to density  actually underlies Kalai et al.'s proof of approximability under log-concave distributions. On the other hand, if the polynomials fail to be dense, then one might conjecture that thresholds cannot be arbitrarily well approximated.

Our result, presented in the next section, confirms the conjecture that even the \emph{sign} function cannot be approximated arbitrarily well by polynomials under \LOL\ distributions.

\subsection{Lower Bound for One Variable} \label{sec:lowerbound}

Consider the \LOL\ density function
\[w_\gamma(x) := C(\gamma) \exp(- |x|^\gamma)\]
on the reals for $\gamma \in (0, 1)$, where $C(\gamma)$ is a normalizing constant. Define the sign function $\sgn(x) = 1$ if $x \ge 0$ and $\sgn(x) = -1$ otherwise. In this section, we show that for sufficiently small $\eps$, the sign function does not have an $L_1$ approximation under the distribution $w_\gamma$. More formally,

\begin{proposition} \label{prop:noapproxforsign}
For any $\gamma \in (0, 1)$, there exists an $\eps = \eps(\gamma)$ such that for any polynomial $p$,
\[\int_\R |p(x) - \sgn(x)| w_\gamma(x) \ dx > \eps.\]
\end{proposition}

The proof is based on the following Markov-type inequality, which roughly says that a bounded polynomial cannot have a large derivative (under the weight $w_\gamma$). This implies the claim, since the sign function we are trying to approximate has a large ``jump'' at the origin.

\begin{lemma} \label{lem:weighted-markov}
For $\gamma \in (0, 1)$ there is a constant $M(\gamma)$ such that 
\[ \sup_{x \in \mathbb{R}} (|p'(x)| w_\gamma(x))  \le M(\gamma) \int_\R |p(x)| w_\gamma(x) \ dx.\]
\end{lemma}

\begin{proof}
The lemma is a combination of a Markov-type inequality and a Nikolskii-type, available in a survey of Nevai \cite{Nevai86}:

\begin{theorem}[\cite{NevaiTo86}, {\cite[Theorem 4.17.4]{Nevai86}}]
There exists a constant $C_1(\gamma)$ such that for any polynomial $p$,
\[\int_{\R} |p'(x)| w_\gamma(x) \ dx \le C_1(\gamma) \int_{\R} |p(x)| w_\gamma(x) \ dx.\]
\end{theorem}

\begin{theorem}[\cite{NevaiTo87}, {\cite[Theorem 4.17.5]{Nevai86}}]
There exists a constant $C_2(\gamma)$ such that for any polynomial $p$,
\[\sup_{x} (|p(x)| w_\gamma(x)) \le C_2(\gamma) \int_{\R} |p(x)| w_\gamma(x) \ dx.\]
\end{theorem}

\end{proof}

\begin{proof}[Proof of Proposition \ref{prop:noapproxforsign}]
Fix $\eps \in (0, 1)$ and suppose $p$ is a polynomial satisfying
\[\int_\R |p(x) - \sgn(x)| w_\gamma(x) \ dx \le \eps.\]
Since the absolute value of the sign function integrates to $1$, this forces
\[\int_\R |p(x)| w_\gamma(x) \ dx \le 1 + \eps \le 2.\]
Therefore, we have by Lemma \ref{lem:weighted-markov} that $|p'(x)| w_\gamma(x) \le 2M(\gamma)$ for every $x$.

The idea is now to show that there is some $x_0$ for which $|p'(x_0)|w_\gamma(x_0) \ge \Omega(1/\eps)$. To see this, let $\delta = 4\eps / C(\gamma)$ and observe that there must exist some $x_+ \in [0, \delta]$ such that $p(x_+) \ge 1/2$. If this were not the case, then we would have
\[\int_\R |p(x)| w_\gamma(x) \ dx \ge \frac{1}{2}\int_0^\delta C(\gamma) \exp(-\delta^\gamma) \ge \eps\]
for $\delta$ small enough to make $\exp(-\delta^\gamma) \ge 1/2$, yielding a contradiction. A similar argument shows that there is some $x_- \in [-\delta, 0]$ with $p(x_-) \le -1/2$. Therefore, by the mean value theorem, there is some $x_0 \in [x_-, x_+]$ with $p'(x_0) \ge 1/2\delta = C(\gamma) / 8\eps$. Moreover, because we took $\delta$ small enough, we also have $p'(x_0) w(x_0) \ge C(\gamma) / 16\eps$. This shows that no polynomial $\eps$-approximates $\sgn$ as long as $\eps < C/32M$.
\end{proof}

Moreover, the proposition shows that it is impossible to get arbitrarily close polynomial approximations to halfspaces under densities $w$ for which there are constants $C$ and $\gamma \in (0, 1)$ with $w(x) \ge C\exp(-|x|^\gamma)$ for all $x \in \R$. This shows that \LOL\ distributions on $\R$ do not support polynomial approximations to halfspaces.

\subsection{Extending the Lower Bound to Multivariate Distributions}  \label{sec:multi-lowerbound}

It is straightforward to extend the lower bound from the previous section to product distributions with \LOL\ marginals.

\begin{theorem}
Let $X = (X_1, \dots, X_n)$ be a random variable over $\R^n$ with density $f_X(x) = w(x_1)f(x_2, \dots, x_n)$. Suppose the density $w$ specifies a univariate $\gamma$-\LOL\ distribution. Then there exists an $\eps = \eps(\gamma)$ such that for any polynomial $p$,
\[\int_{\R^n} |p(x_1, \dots, x_n) - \sgn(x_1)| f_X(x_1, \dots, x_n) \ dx_1 dx_2 \dots dx_n > \eps.\]
That is, the linear threshold function $\sgn(x_1)$ cannot be approximated arbitrarily well by polynomials.
\end{theorem}

\begin{proof}
Let $p(x_1, \dots, x_n)$ be a polynomial, and define a univariate polynomial $q$ by ``averaging out'' the variables $x_2, \dots, x_n$:
\[q(x_1) := \int_{\R^{n-1}} p(x_1, \dots, x_n) f(x_2, \dots, x_n) \ dx_2 \dots \ dx_n.\]
Then we have
\begin{align*}
\int_{\R} |q(x_1) - \sgn(x_1)| w(x_1) \ dx_1 &= \int_{\R} \left|\int_{\R^{n-1}} (p(x_1, \dots, x_n) - \sgn(x_1)) f(x_2, \dots, x_n) \ dx_2 \dots dx_n\right| w(x_1) \ dx_1 \\
&\le \int_{\R} \left(\int_{\R^{n-1}} |p(x_1, \dots, x_n) - \sgn(x_1)| f(x_2, \dots, x_n) \ dx_2 \dots dx_n \right)w(x_1) \ dx_1 \\
&= \int_{\R^n} \int_{\R^n} |p(x_1, \dots, x_n) - \sgn(x_1)| f_X(x_1, \dots, x_n) \ dx_1 dx_2 \dots dx_n.
\end{align*}
By Proposition \ref{prop:noapproxforsign}, the latter quantity must be at least $\eps(\gamma)$.
\end{proof}

Let $w_\gamma^n(x) \propto \exp(-(|x_1|^\gamma + \dots + |x_n|^\gamma))$ denote the density of the “prototypical” multivariate LSL distribution, with each marginal having the same exponential power law distribution. Our impossibility result holds uniformly for every distribution in the sequence $\{w_\gamma^n\}$. That is, for every $\gamma \in (0, 1)$, there exists $\eps = \eps(\gamma)$ for which halfspaces cannot be learned by polynomials under any of the distributions specified by $\{w_\gamma^n\}$.

As a consequence, we get inapproximability results for several natural classes of distributions that dominate $\{w_\gamma^n\}$ by constant factors (i.e. not growing with $n$). 


\begin{enumerate}
\item Any power-law distribution, i.e. a distribution with density $\propto \|x\|^{-M}$ for some constant $M$, since such a distribution dominates every $w^n_\gamma$.
\item Multivariate generalizations of the log-normal distribution, i.e. any distribution with density $\propto \exp(-\operatorname{polylog}(\|x\|))$.
\item Multivariate exponential power distributions, which have densities $\propto \exp(-\|x\|^\gamma)$ for $\gamma \in (0, 1)$. These distributions dominate the prototypical $w^n_\gamma$ by the inequality of $\ell_p$-norms:
\[\|x\|^\gamma \le |x_1|^\gamma + \dots + |x_n|^\gamma\]
for every $0 \le \gamma \le 2$.

\end{enumerate}

\section{Tail Bounds for Limited Independence}
Our proof consists of three steps:
\begin{itemize}
\item[\S \ref{sec:dual}] First we reformulate the question of tail bounds for $k$-wise independent distributions using linear programming duality and symmetrisation. This reduces the problem to proving a degree lower bound on univariate polynomials. Namely we need to give a lower bound on the degree of a polynomial $p : \{0,1, \cdots n\} \to \mathbb{R}$ such that $p(i) \geq 0$ for all $i$, $p(i) \geq 1$ if $|i-n/2|\geq T$, and $\ex{}{p(i)} \leq \delta$, where $i$ is drawn from the binomial distribution.
\item[\S \ref{sec:cts}] We then transform the problem from one about polynomials with a discrete domain to one about polynomials with a continuous domain. This amounts to showing that, since $\ex{}{p(i)} \leq \delta$ with respect to the binomial distribution, we can bound $\ex{}{p(x+n/2)}$ with respect to a truncated Gaussian distribution on $x$. 
\item[\S \ref{sec:kwiselb}] Finally we can apply the tools of weighted approximation theory. We know that $p(x+n/2)$ is small for $x$ near the origin, but $p(T+n/2) \geq 1$. We show that any low-degree polynomial that is bounded near the origin cannot grow too quickly. This implies that $p$ must have high degree.
\end{itemize}

\subsection{Dual Formulation} \label{sec:dual}
Question \ref{q:Kwise} from the introduction is equivalent to finding the smallest $k$ for which the value of the following linear program is at most $\delta$.
\begin{center}Linear Program Formulation of Question \ref{q:Kwise}\end{center}
\begin{align*}
\max_{\psi} &\sum_{x \in \{-1, 1\}^n} \psi(x) F_T(x) \\
\text{s.t.} &\sum_{x \in \{-1, 1\}^n} \psi(x) \chi_S(x) = 0 & \text{for all } |S| \le k \\
&\sum_{x \in \{-1, 1\}^n} \psi(x) = 1 \\
& 0 \le \psi(x) \le 1 & \text{for all } x \in \{-1, 1\}^n.
\end{align*}
Here, $F_T(x) = 1$ if $|x| \ge T$ and is $0$ otherwise, and $\chi_S(x)$ is the Fourier character corresponding to $S \subseteq [n]$. 

If we set $\pr{X}{X=x} = \phi(x)$, then the constraints impose that $X$ is a $k$-wise independent distribution, while the objective function is $\pr{X}{\left| \sum_{i \in [n]} X_i \right| \geq T}$. Thus the above linear program finds the $k$-wise independent distribution with the worst tail bound. If the value of the program is at most $\delta$, then all $k$-wise independent distributions satisfy the tail bound, as required.

Taking the dual of the above linear program yields the following.
\begin{center}Dual Formulation of Question \ref{q:Kwise}\end{center}
\begin{align*}
\min_{p} &\ 2^{-n}\sum_{x \in \{-1, 1\}^n} p(x)\\
\text{s.t.} &\ \text{deg}(p) \leq k\\
&\ p(x) \ge F_T(x) & \text{for all } x \in \{-1, 1\}^n.
\end{align*}
By strong duality, the value of the dual linear program is the same as that of the primal.

The multilinear polynomial $p$ as an ``upper sandwich'' of $F_T$ -- that is, $p \geq F_T$ and $\ex{X \in \{\pm 1\}^n}{p(X)}$ is minimal. Therefore, $k(n, \delta, T)$ is the smallest $k$ for which $F_T$ admits an upper sandwiching polynomial of degree $k$ with expectation $\delta$.

Consider the shifted univariate symmetrization of $F_T$
\[F_T'(x) = \begin{cases}
1 \text{ if } |x - n/2| \ge T \\
0 \text{ otherwise.} 
\end{cases}\]
By applying the well-known Minsky-Papert symmetrization \cite{MinskyPapert} to the dual formulation above, we get the following characterization.

\begin{theorem} \label{thm:discrete-char}
The quantity $k(n, \delta, T)$ from Question \ref{q:Kwise} is the smallest $k$ for which there exists a degree-$k$ univariate polynomial $p : \{0, \dots, n\} \to \R$ such that
\begin{enumerate}
\item $p(i) \ge F_T'(i)$ for all $0 \le i \le n$ and
\item $2^{-n}\sum_{i = 0}^n {n \choose i} p(i) \le \delta$.
\end{enumerate}
\end{theorem}

The upper bound on $k(n,\delta,T)$ (Theorem \ref{thm:KwiseUpperBound}) is proved (in the appendix) by showing that $$p(i) = \left(\frac{i-n/2}{T}\right)^k$$
satisfies the requirements of Theorem \ref{thm:discrete-char} for an appropriate even $k$.\footnote{While our results show that this polynomial is \emph{asymptotically} optimal, numerical experiments have shown that it is not exactly optimal.} So this characterisation does in fact capture how upper bounds are proved. The fact that it is a tight characterisation allows us to prove that a barrier to the technique is in fact an impossibility result.


With this characterisation of our problem, we may move on to proving inapproximability results.

\subsection{A Continuous Version} \label{sec:cts}

To apply techniques from the theory of weighted polynomial approximations, we move to polynomials on a continuous domain. We replace the binomial distribution upon which Theorem \ref{thm:discrete-char} evaluates $p$ with a Gaussian distribution.

Define the probability density function
\[w(x) = \frac{1}{\sqrt{\pi}}e^{-x^2}.\]
We define the $L_\infty$ norm with respect to the weight $w$:
\[\|g\|_{L_\infty(S)} = \sup_{x \in S} |g(x)|w(x).\]

Now we can give the continuous version of the problem:

\begin{theorem} \label{thm:approximation}
Let $T = c \sqrt{n \log(1/\delta)}$ for $c \ge 5$, and $d = k(n, \delta, T)$. Assume $n \geq (12c)^2 (3 \log(1/\delta))^3$. Then for $T' = 4cT/\sqrt{n}$, there is a degree $d$ polynomial $q$ such that
\begin{enumerate}
\item $q(T') = q(-T') \ge 1$ and
\item $\|q\|_{L_\infty[-\sqrt{d}, \sqrt{d}]} \le \delta^{0.9}(n+1)$.
\end{enumerate}
\end{theorem}

The following lemma is key to moving from the discrete to the continous setting. It shows that if a polynomial is bounded at evenly spaced points, then it must also be bounded between those points, assuming the number of points is sufficiently large relative to the degree.

\begin{lemma} \cite{EhlichZe64, RivlinCh66, NisanSz94} \label{lem:Continous}
Let $q$ be a polynomial of degree $d$ such that $|q(i)| \le 1$ for $i = 0, 1, \dots, m$, where $3d^2 \le m$. Then $|q(x)| \le \frac{3}{2}$ for all $x \in [0, m]$.
\end{lemma}

\begin{proof}
Let $a = \max_{x \in [0, m]} |q'(x)|$. Then by the mean value theorem, $|q(x)| \le 1 + a/2$ for $x \in [0, m]$. By Markov's inequality (\cite{Markov90}, see also \cite{Cheney82}),
\[a \le \frac{2d^2 (1 + a/2)}{m}.\]
Rearranging gives
\[\frac{a}{2 + a} \le \frac{d^2}{m} \le \frac{1}{3}.\]
Therefore, $a \le 1$, and hence $|q(x)| \le \frac{3}{2}$ for $x \in [0, m]$.
\end{proof}

We also require the following anti-concentration lemma.

\begin{lemma} \label{lem:BinomialBound}
\[{n \choose n/2 + \alpha \sqrt{n}} \ge \frac{2^{n - 6\alpha^2}}{n + 1}.\]
\end{lemma}

\begin{proof}
It is well known via Stirling's approximation that ${n \choose k} \ge 2^{nH(k/n)}/(n+1)$, where $H( \cdot )$ denotes the binary entropy function. We estimate
\begin{align*}
H\left(\frac{1}{2} + \frac{\alpha}{\sqrt{n}}\right) &\ge \left(\frac{1}{2} + \frac{\alpha}{\sqrt{n}}\right)\left(1 - \frac{2\alpha}{(\log 2)\sqrt{n}}\right) + \left(\frac{1}{2} - \frac{\alpha}{\sqrt{n}}\right)\left(1 + \frac{2\alpha}{(\log 2)\sqrt{n}}\right) \\
&\ge 1 - \frac{4\alpha^2}{(\log 2) n},
\end{align*}
which concludes the proof.
\end{proof}

\begin{proof}[Proof of Theorem \ref{thm:approximation}]
Let $p$ be the polynomial promised by Theorem \ref{thm:discrete-char}. By Theorem \ref{thm:KwiseUpperBound}, we know that  $d \leq 3 \log(1/\delta)$. Define
\[q(x) = p(x\sqrt{n}/4c+n/2).\]
Then $q(\pm T') = p(\pm T+n/2) \ge F_T(\pm T + n/2)= 1$, dispensing with the first claim.

Now for all integers $i$ in the interval $n/2 \pm \sqrt{nd}/4c$, we have
\[2^{-n} {n \choose i} |p(i)| \le \delta\]
and hence, by Lemma \ref{lem:BinomialBound},
\[|p(i)| \le \frac{2^n \delta}{ {n\choose n/2 + \sqrt{nd}/4c} } \le (n+1)\delta 2^{6d/16c^2} \leq (n+1) \delta^{1 - 18/16c^2} \leq (n+1) \delta^{0.9}. \]
By Lemma \ref{lem:Continous}, $|p(x)| \le \frac{3}{2}(n+1)\delta^{0.9}$ on the whole interval $n/2 \pm \sqrt{nd}/4c$. Thus $|q(x)| \le \frac{3}{2}(n+1)\delta^{0.9}$ on $[-\sqrt{d}, \sqrt{d}]$, completing the proof.
\end{proof}

\subsection{The Lower Bound} \label{sec:kwiselb}

Now we state the result we need from approximation theory. The following ``infinite-finite range inequality'' shows that the norm of weighted polynomial on the real line is determined by its norm on a finite interval around the origin. Thus, an upper bound on the magnitude of a polynomial near the origin yields a bound on its growth away from the origin.. We will apply this to the polynomial given to us in Theorem \ref{thm:approximation}.

\begin{theorem} \label{thm:homegrown-infinite-finite}
For any polynomial $p$ of degree $d$ and $B > 1$,
\[\|p\|_{L_\infty(\R \setminus [-B\sqrt{d}, B\sqrt{d}])} \le (2eB)^d \exp(-B^2d)\|p\|_{L_\infty[-\sqrt{d}, \sqrt{d}]}.\]
\end{theorem}

The proof follows \cite[Theorem 6.1]{Lubinsky07} and \cite[Theorem 4.16.12]{Nevai86}.
\begin{proof}
Let $\tilde{p}$ be a polynomial of degree $d$. Let $T_d(x)$ denote the $d$th Chebyshev polynomial of the first kind \cite{Cheney82}. By the extremal properties of $T_d$, we have
\[|\tilde{p}(x)| \le |T_d(x)| \left(\max_{t \in [-1, 1]}|\tilde{p}(t)|\right) \le (2|x|)^d \left(\max_{t \in [-1, 1]}|\tilde{p}(t)|\right) \]
for $|x| \ge 1$. Rescaling $p(x)=\tilde{p}(x/\sqrt{d})$ yields
\[|p(x)| \le \left(\frac{2|x|}{\sqrt{d}}\right)^d \left(\max_{t \in [-\sqrt{d}, \sqrt{d}]} |p(t)|\right) \le \sqrt{\pi}e^d\left(\frac{2|x|}{\sqrt{d}}\right)^d \|p\|_{L_\infty[-\sqrt{d}, \sqrt{d}]}\]
for $|x| \ge \sqrt{d}$. Now let $|x| = B\sqrt{d}$ for some $B > 1$. Then
\[|p(x)|w(x) \le e^d(2B)^d \exp(-B^2d)\|p\|_{L_\infty[-\sqrt{d}, \sqrt{d}]}.\]
Since the coefficient $(2eB)^d\exp(-B^2d)$ is decreasing in $B$, this proves the claim.
\end{proof}

The above approximation theory result, combined with our continuous formulation Theorem \ref{thm:approximation}, enables us to complete the proof.

\begin{theorem} \label{thm:deg-lb}
Let $T = c \sqrt{n \log(1/\delta)}$ for $c \ge 5$. Assume $n \geq (12c)^2 (3 \log(1/\delta))^3$.  Then $k(n,\delta,T) > \log(1/\delta)/9\log c$.
\end{theorem}

\begin{proof}
Let $q$ be the polynomial given by Theorem \ref{thm:approximation}. Let $T'=4cT/\sqrt{n}$, $d = \log(1/\delta)/9\log c$, and $B=T'/\sqrt{d}=12c^2 \sqrt{\log c}$. For the sake of contradiction, we suppose that $q$ satisfies the conditions of Theorem \ref{thm:approximation}, but $\deg(q) \le d$. Then
\[\|q\|_{L_\infty(\R \setminus [-B\sqrt{d}, B\sqrt{d}])} = \|q\|_{L_\infty(\R \setminus [-T', T'])} \ge \frac{\exp(-T'^2)}{\sqrt{\pi}}.\]
On the other hand, applying Theorem \ref{thm:homegrown-infinite-finite}, gives
$$\|q\|_{L_\infty(\R \setminus [-B\sqrt{d}, B\sqrt{d}])} \leq (2eB)^d \exp(-T'^2) \delta^{0.9} (n+1).$$
Combining the two inequalities gives
$$\frac{1}{\sqrt{\pi}} \leq (2eB)^d \delta^{0.9} (n+1) \leq \left(24 e c^2 \sqrt{\log(c)}\right)^{\log(1/\delta)/9\log(c)} \delta^{0.9} (n+1) \leq \delta^{1/3} (n+1),$$ which is a contradiction.
\end{proof}

Theorem \ref{thm:deg-lb} yields Theorem \ref{thm:KwiseLowerBound}.


\section{Further Work}
Our negative results naturally suggest a number of directions for future work.

Are halfspaces agnostically learnable under \LOL\ distributions? Our negative result does not even necessarily rule out the use of $L_1$ regression for this task: The polynomial regression algorithm of Kalai et al. \cite{KalaiKlMaSe08} is in fact quite flexible. Nothing is really special about the basis of low-degree monomials, and the algorithm works equally well over any small, efficiently evaluable ``feature space''. That is, if we can show that halfspaces are well-approximated by linear combinations of features from a feature space $\mathcal{F}$ under a distribution $\mathcal{D}$, then we can agnostically learn halfspaces with respect to $\mathcal{D}$ in time proportional to $|\mathcal{F}|$. Could one hope for such approximations? Wimmer \cite{Wimmer10} and Feldman and Kothari \cite{FeldmanKo14} have shown how to use non-polynomial basis functions to obtain faster learning algorithms on the boolean hypercube. On the other hand, recent work of Dachman-Soled et al. \cite{DachmanFeTaWaWi14} shows that, at least for product distributions on the hypercube, polynomials yield the best basis for $L_1$ regression.

Are there other suitable derandomizations of concentration inequalities? In this work, we focused on understanding the limits of $k$-wise independent distributions. Gopalan et al. \cite{GopalanKM14} gave a much more sophisticated generator with nearly optimal seed length. But could simple, natural pseudorandom distributions, such as small-bias spaces, give strong tail bounds themselves?

\section{Acknowledgements}

We thank Varun Kanade, Scott Linderman, Raghu Meka, Jelani Nelson, Justin Thaler, Salil Vadhan, Les Valiant, and several anonymous reviewers for helpful discussions and comments.

\bibliographystyle{alpha}
\bibliography{log_convex}

\newcommand{\etalchar}[1]{$^{#1}$}
\begin{thebibliography}{GOWZ10}

\bibitem[ABI85]{ABI85}
Noga Alon, Laszlo Babai, and Alon Itai.
\newblock A fast and simple randomized parallel algorithm for the maximal
  independent set problem.
\newblock Technical report, Chicago, IL, USA, 1985.

\bibitem[AS04]{AaronsonSh04}
Scott Aaronson and Yaoyun Shi.
\newblock Quantum lower bounds for the collision and the element distinctness
  problems.
\newblock {\em J. ACM}, 51(4):595--605, 2004.

\bibitem[Baz09]{Bazzi}
Louay M.~J. Bazzi.
\newblock Polylogarithmic independence can fool {DNF} formulas.
\newblock {\em SIAM J. Comput.}, 38(6):2220--2272, March 2009.

\bibitem[BBC{\etalchar{+}}01]{BealsBuClMoWo01}
Robert Beals, Harry Buhrman, Richard Cleve, Michele Mosca, and Ronald de~Wolf.
\newblock Quantum lower bounds by polynomials.
\newblock {\em J. ACM}, 48(4):778--797, 2001.

\bibitem[Bei93]{Beigel93}
Richard Beigel.
\newblock The polynomial method in circuit complexity.
\newblock In {\em Structure in Complexity Theory Conference}, pages 82--95.
  IEEE Computer Society, 1993.

\bibitem[Bei94]{Beigel94}
Richard Beigel.
\newblock Perceptrons, {PP}, and the polynomial hierarchy.
\newblock {\em Computational Complexity}, 4:339--349, 1994.

\bibitem[Ber24]{Bernstein24}
S.~N. Bernstein.
\newblock Le probl{\`e}me de l'approximation des fonctions continues sur tout
  l'axe r{\'e}el et l'une de ses applications.
\newblock {\em Bull. Math. Soc. France}, 52:399€"--410, 1924.

\bibitem[BFJ{\etalchar{+}}94]{BlumFuJaKeMaRu94}
Avrim Blum, Merrick Furst, Jeffrey Jackson, Michael Kearns, Yishay Mansour, and
  Steven Rudich.
\newblock Weakly learning {DNF} and characterizing statistical query learning
  using {F}ourier analysis.
\newblock In {\em Proceedings of the twenty-sixth annual ACM symposium on
  Theory of computing}, page 253Ð262. ACM, 1994.

\bibitem[Bon70]{Bonami70}
Aline Bonami.
\newblock {\'E}tude des coefficients de fourier des fonctions de $l^p(g)$.
\newblock {\em Annales de l'institut Fourier}, 20(2):335--402, 1970.

\bibitem[BOW10]{BlaisODWi10}
Eric Blais, Ryan O'Donnell, and Karl Wimmer.
\newblock Polynomial regression under arbitrary product distributions.
\newblock {\em Machine Learning}, 80(2-3):273--294, 2010.

\bibitem[BR94]{BRchernoff}
M.~Bellare and J.~Rompel.
\newblock Randomness-efficient oblivious sampling.
\newblock In {\em FOCS}, pages 276--287, Nov 1994.

\bibitem[BRRY10]{BRRY}
Mark Braverman, Anup Rao, Ran Raz, and Amir Yehudayoff.
\newblock Pseudorandom generators for regular branching programs.
\newblock {\em FOCS}, pages 40--47, 2010.

\bibitem[BV10]{BrodyVerbin}
Joshua Brody and Elad Verbin.
\newblock The coin problem and pseudorandomness for branching programs.
\newblock In {\em FOCS}, pages 30--39, 2010.

\bibitem[Car51]{Carleson51}
Lennart Carleson.
\newblock Bernstein's approximation problem.
\newblock {\em Proc. Amer. Math. Soc.}, 2:953--961, 1951.

\bibitem[Che82]{Cheney82}
E.W. Cheney.
\newblock {\em Introduction to Approximation Theory}.
\newblock AMS Chelsea Publishing Series. AMS Chelsea Pub., 1982.

\bibitem[DETT10]{DETT}
Anindya De, Omid Etesami, Luca Trevisan, and Madhur Tulsiani.
\newblock Improved pseudorandom generators for depth 2 circuits.
\newblock In Maria Serna, Ronen Shaltiel, Klaus Jansen, and Jos� Rolim,
  editors, {\em Approximation, Randomization, and Combinatorial Optimization.
  Algorithms and Techniques}, volume 6302 of {\em Lecture Notes in Computer
  Science}, pages 504--517. 2010.

\bibitem[DFT{\etalchar{+}}14]{DachmanFeTaWaWi14}
Dana Dachman{-}Soled, Vitaly Feldman, Li{-}Yang Tan, Andrew Wan, and Karl
  Wimmer.
\newblock Approximate resilience, monotonicity, and the complexity of agnostic
  learning.
\newblock {\em CoRR}, abs/1405.5268, 2014.
\newblock To appear in SODA 2015.

\bibitem[DGJ{\etalchar{+}}09]{DGJSV09}
Ilias Diakonikolas, Parikshit Gopalan, Ragesh Jaiswal, Rocco~A. Servedio, and
  Emanuele Viola.
\newblock Bounded independence fools halfspaces.
\newblock In {\em In Proc. 50th Annual Symposium on Foundations of Computer
  Science (FOCS)}, pages 171--180, 2009.

\bibitem[DLS14]{DanielyLiSh14}
Amit Daniely, Nati Linial, and Shai Shalev{-}Shwartz.
\newblock The complexity of learning halfspaces using generalized linear
  methods.
\newblock {\em CoRR}, abs/1211.0616, 2014.

\bibitem[DSTW10]{DSTW}
Ilias Diakonikolas, Rocco~A. Servedio, Li-Yang Tan, and Andrew Wan.
\newblock A regularity lemma, and low-weight approximators, for low-degree
  polynomial threshold functions.
\newblock In {\em Proceedings of the 2010 IEEE 25th Annual Conference on
  Computational Complexity}, CCC '10, pages 211--222, Washington, DC, USA,
  2010. IEEE Computer Society.

\bibitem[EZ64]{EhlichZe64}
H.~Ehlich and K.~Zeller.
\newblock Schwankung von polynomen zwischen gitterpunkten.
\newblock {\em Mathematische Zeitschrift}, 86:41--44, 1964.

\bibitem[FGKP06]{FeldmanGoKhPo06}
Vitaly Feldman, Parikshit Gopalan, Subhash Khot, and Ponnuswami.
\newblock New results for learning noisy parities and halfspaces.
\newblock In {\em Proceedings of the 47th Annual IEEE Symposium on Foundations
  of Computer Science}, FOCS '06, pages 563--574, Washington, DC, USA, 2006.
  IEEE Computer Society.

\bibitem[FK14]{FeldmanKo14}
Vitaly Feldman and Pravesh Kothari.
\newblock Agnostic learning of disjunctions on symmetric distributions.
\newblock {\em CoRR}, abs/1405.6791, 2014.

\bibitem[FLS11]{FeldmanLeSe11}
V.~Feldman, H.~Lee, and R.~Servedio.
\newblock Lower bounds and hardness amplification for learning shallow monotone
  formulas.
\newblock {\em Journal of Machine Learning Research - COLT Proceedings},
  19:273Ð292, 2011.

\bibitem[GKK08]{GopalanKaKl08}
Parikshit Gopalan, Adam~Tauman Kalai, and Adam~R. Klivans.
\newblock Agnostically learning decision trees.
\newblock In Cynthia Dwork, editor, {\em STOC}, pages 527--536. ACM, 2008.

\bibitem[GKM14]{GopalanKM14}
Parikshit Gopalan, Daniel Kane, and Raghu Meka.
\newblock Pseudorandomness for concentration bounds and signed majorities.
\newblock {\em CoRR}, abs/1411.4584, 2014.

\bibitem[GOWZ10]{GOWZ}
Parikshit Gopalan, Ryan O'Donnell, Yi~Wu, and David Zuckerman.
\newblock Fooling functions of halfspaces under product distributions.
\newblock In {\em Proceedings of the 2010 IEEE 25th Annual Conference on
  Computational Complexity}, CCC '10, pages 223--234, Washington, DC, USA,
  2010. IEEE Computer Society.

\bibitem[GR06]{GuruswamiRa06}
V.~Guruswami and P.~Raghavendra.
\newblock Hardness of learning halfspaces with noise.
\newblock In {\em Proceedings of FOCS Õ06}, page 543Ð552, 2006.

\bibitem[HNO08]{HarveyNeOn08}
Nicholas J.~A. Harvey, Jelani Nelson, and Krzysztof Onak.
\newblock Sketching and streaming entropy via approximation theory.
\newblock In {\em FOCS}, pages 489--498, 2008.

\bibitem[Hoe63]{Hoeffding}
Wassily Hoeffding.
\newblock Probability inequalities for sums of bounded random variables.
\newblock {\em Journal of the American Statistical Association}, 58(301):pp.
  13--30, 1963.

\bibitem[IK37]{IzumiKa37}
S.~Izumi and T.~Kawata.
\newblock Quasi-analytic class and closure of $\{t^n\}$ in the interval
  $(-\infty, \infty)$.
\newblock {\em Tohoku Math. J.}, 43:267--273, 1937.

\bibitem[INW94]{INW}
Russell Impagliazzo, Noam Nisan, and Avi Wigderson.
\newblock Pseudorandomness for network algorithms.
\newblock In {\em STOC}, pages 356--364, 1994.

\bibitem[Kea98]{Kearns98}
M.~Kearns.
\newblock Efficient noise-tolerant learning from statistical queries.
\newblock {\em Journal of the ACM}, 45(6):983Ð1006, 1998.

\bibitem[KKM13]{KaneKlMe13}
Daniel~M. Kane, Adam Klivans, and Raghu Meka.
\newblock Learning halfspaces under log-concave densities: Polynomial
  approximations and moment matching.
\newblock In {\em COLT}, pages 522--545, 2013.

\bibitem[KKMS08]{KalaiKlMaSe08}
Adam~Tauman Kalai, Adam~R. Klivans, Yishay Mansour, and Rocco~A. Servedio.
\newblock Agnostically learning halfspaces.
\newblock {\em SIAM J. Comput.}, 37(6):1777--1805, 2008.

\bibitem[KNP11]{KNP}
Michal Kouck\'{y}, Prajakta Nimbhorkar, and Pavel Pudl\'{a}k.
\newblock Pseudorandom generators for group products.
\newblock In {\em STOC}, pages 263--272, 2011.

\bibitem[KOS08]{KlivansODSe08}
Adam~R. Klivans, Ryan O'Donnell, and Rocco~A. Servedio.
\newblock Learning geometric concepts via gaussian surface area.
\newblock In {\em FOCS}, pages 541--550, 2008.

\bibitem[KS04]{KlivansSe04}
Adam~R. Klivans and Rocco~A. Servedio.
\newblock Learning {DNF} in time $2^{\tilde{o}(n^{1/3})}$.
\newblock {\em J. Comput. Syst. Sci.}, 68(2):303--318, 2004.

\bibitem[KS07]{KlivansSh07}
Adam~R Klivans and Alexander~A Sherstov.
\newblock Unconditional lower bounds for learning intersections of halfspaces.
\newblock {\em Machine Learning}, 69(2-3):97Ð114, 2007.

\bibitem[KS09]{KlivansSh09}
Adam~R. Klivans and Alexander~A. Sherstov.
\newblock Cryptographic hardness for learning intersections of halfspaces.
\newblock {\em J. Comput. Syst. Sci.}, 75(1):2Ð12, 2009.

\bibitem[KS10]{KlivansSh10}
Adam~R. Klivans and Alexander~A. Sherstov.
\newblock Lower bounds for agnostic learning via approximate rank.
\newblock {\em Computational Complexity}, 19(4):581--604, 2010.

\bibitem[KSSH94]{KearnsScSeHe94}
Michael Kearns, Robert~E. Schapire, Linda~M. Sellie, and Lisa Hellerstein.
\newblock Toward efficient agnostic learning.
\newblock In {\em Machine Learning}, pages 341--352. ACM Press, 1994.

\bibitem[LBW95]{LeeBaWi95}
Wee~Sun Lee, Peter~L. Bartlett, and Robert~C. Williamson.
\newblock On efficient agnostic learning of linear combinations of basis
  functions.
\newblock In {\em Proceedings of the Eighth Annual Conference on Computational
  Learning Theory}, COLT '95, pages 369--376, New York, NY, USA, 1995. ACM.

\bibitem[Lub07]{Lubinsky07}
Doron Lubinsky.
\newblock A survey of weighted polynomial approximation with exponential
  weights.
\newblock {\em Surveys in Approximation Theory}, 3:1--105, 2007.

\bibitem[Mar90]{Markov90}
A.~A. Markov.
\newblock On a question of {D}. {I}. {M}endeleev.
\newblock {\em Zapiski Imperatorskoi Akademii Nauk,}, 62:1--24, 1890.

\bibitem[MP72]{MinskyPapert}
Marvin Minsky and Seymour Papert.
\newblock {\em Perceptrons: An Introduction to Computational Geometry}.
\newblock MIT Press, Cambridge MA, 1972.

\bibitem[MZ10]{MekaZuckerman}
Raghu Meka and David Zuckerman.
\newblock Pseudorandom generators for polynomial threshold functions.
\newblock In {\em Proceedings of the Forty-second ACM Symposium on Theory of
  Computing}, STOC '10, pages 427--436, New York, NY, USA, 2010. ACM.

\bibitem[Nev86]{Nevai86}
Paul Nevai.
\newblock G\'{e}za {F}reud, orthogonal polynomials and {C}hristoffel functions.
  {A} case study.
\newblock {\em Journal of Approximation Theory}, 48(1):3--167, 1986.

\bibitem[Nis92]{Nisan}
Noam Nisan.
\newblock $\mathcal{RL}\subset\mathcal{SC}$.
\newblock In {\em STOC}, pages 619--623, 1992.

\bibitem[NN93]{NaorN93}
Joseph Naor and Moni Naor.
\newblock Small-bias probability spaces: Efficient constructions and
  applications.
\newblock {\em SIAM J. Computing}, 22:838--856, 1993.

\bibitem[NS94]{NisanSz94}
N.~Nisan and M.~Szegedy.
\newblock On the degree of boolean functions as real polynomials.
\newblock {\em Computational Complexity}, 4:301--313, 1994.

\bibitem[NT86]{NevaiTo86}
Paul Nevai and Vilmos Totik.
\newblock Weighted polynomial inequalities.
\newblock {\em Constructive Approximation}, 2(1):113--127, 1986.

\bibitem[NT87]{NevaiTo87}
P.~Nevai and V.~Totik.
\newblock Sharp {N}ikolskii inequalities with exponential weights.
\newblock {\em Analysis Mathematica}, 13(4):261--267, 1987.

\bibitem[O'D14]{O'Donnell14}
Ryan O'Donnell.
\newblock {\em Analysis of Boolean Functions}.
\newblock Cambridge University Press, New York, NY, USA, 2014.

\bibitem[Pat92]{Paturi92}
Ramamohan Paturi.
\newblock On the degree of polynomials that approximate symmetric boolean
  functions (preliminary version).
\newblock In S.~Rao Kosaraju, Mike Fellows, Avi Wigderson, and John~A. Ellis,
  editors, {\em STOC}, pages 468--474. ACM, 1992.

\bibitem[RC66]{RivlinCh66}
T.~J. Rivlin and E.~W. Cheney.
\newblock A comparison of uniform approximations on an interval and a finite
  subset thereof.
\newblock {\em SIAM J. Numer. Anal.}, 3(2):311--320, 1966.

\bibitem[Rei08]{Reingold}
Omer Reingold.
\newblock Undirected connectivity in log-space.
\newblock {\em J. ACM}, 55(4):17:1--17:24, September 2008.

\bibitem[RSV13]{RSV13}
Omer Reingold, Thomas Steinke, and Salil Vadhan.
\newblock Pseudorandomness for regular branching programs via fourier analysis.
\newblock In {\em APPROX-RANDOM}, pages 655--670, 2013.

\bibitem[She08]{Sherstov08}
Alexander~A. Sherstov.
\newblock Communication lower bounds using dual polynomials.
\newblock {\em Bulletin of the EATCS}, 95:59--93, 2008.

\bibitem[She09]{Sherstov09}
Alexander~A. Sherstov.
\newblock Separating {AC$^{\mbox{0}}$} from depth-2 majority circuits.
\newblock {\em SIAM J. Comput.}, 38(6):2113--2129, 2009.

\bibitem[SSS95]{SSS}
J.~Schmidt, A.~Siegel, and A.~Srinivasan.
\newblock {Chernoff--Hoeffding} bounds for applications with limited
  independence.
\newblock {\em SIAM J. Discrete Mathematics}, 8(2):223--250, 1995.

\bibitem[SSSS11]{ShalevShSr11}
S.~Shalev-Shwartz, O.~Shamir, and K.~Sridharan.
\newblock Learning kernel-based halfspaces with the 0-1 loss.
\newblock {\em SIAM Journal on Computing}, 40:1623Ð1646, 2011.

\bibitem[SV14]{SachdevaVi14}
Sushant Sachdeva and Nisheeth~K. Vishnoi.
\newblock Faster algorithms via approximation theory.
\newblock {\em Foundations and Trends in Theoretical Computer Science},
  9(2):125--210, 2014.

\bibitem[Val84]{Valiant84}
Leslie~G. Valiant.
\newblock A theory of the learnable.
\newblock {\em Commun. ACM}, 27(11):1134--1142, 1984.

\bibitem[Wim10]{Wimmer10}
Karl Wimmer.
\newblock Agnostically learning under permutation invariant distributions.
\newblock In {\em Proceedings of the 2010 IEEE 51st Annual Symposium on
  Foundations of Computer Science}, FOCS '10, pages 113--122, Washington, DC,
  USA, 2010. IEEE Computer Society.

\end{thebibliography}

\appendix
\section{Upper Bound for Limited Independence}

Theorem \ref{thm:KwiseUpperBound} follows from the following well-known \cite{SSS,BRchernoff} lemma, which we prove for completeness. 
\begin{lemma}\label{lem:MomentBound}
Let $X  \in \{\pm 1\}^n$ be uniform and $r \in \mathbb{R}^n$. For all even $k \geq 2$, $$\ex{}{(X \cdot r)^k} \leq \left( e \norm{r}_2^2 k \right)^{k/2}.$$
\end{lemma}
An even stronger form of Lemma \ref{lem:MomentBound} follows immediately from the hypercontractivity theorem \cite{Bonami70} \cite[\S 9]{O'Donnell14}: Letting $f(x) = x \cdot r$, we have $$\ex{}{(X \cdot r)^k}  = \norm{f}_k^k \leq \left((k-1)^{\deg(f)/2} \norm{f}_2\right)^k =  \left(\sqrt{k-1} \norm{r}_2\right)^k,$$ as required. A self-contained proof follows.
\begin{proof}
We start by bounding the moment generating function of $X \cdot r$:
Let $t \in \mathbb{R}$ be fixed later. For any $i \in [n]$, we have $$\ex{}{e^{t r_i X_i}} = \frac12 \left( e^{t r_i} + e^{-t r_i} \right) = \sum_{k=0}^\infty \frac{(t r_i)^k + (-t r_i)^k}{2 k!} = \sum_{k=0}^\infty \frac{(t r_i)^{2k}}{(2k)!} \leq \sum_{k=0}^\infty \frac{(t^2 r_i^2)^{k}}{2^k k!} = e^{t^2 r_i^2 /2}.$$
By independence, $$\ex{}{e^{t (X \cdot r)}} = \prod_{i=1}^n \ex{}{e^{t r_i X_i}} \leq \prod_{i=1}^n e^{t^2 r_i^2/2} = e^{t^2 \norm{r}_2^2/2}.$$
Now we have $$\ex{}{e^{t (X \cdot r)}} = \sum_{k=0}^\infty \frac{t^k}{k!} \ex{}{(X \cdot r)^k} \leq e^{t^2 \norm{r}_2^2/2}.$$
We wish to bound a single moment, namely $\ex{}{(X \cdot r)^{k_*}}$ for an even $k_*$. We do this by picking one term out of the above infinite sum.
We have  $\ex{}{(X \cdot r)^k} \geq 0$ for even $k$, so these terms can be removed from the sum without increasing it. By changing the sign of $t$, we can ensure that the sum of the odd terms is positive and thus $$\frac{t^{k_*}}{k_*!} \ex{}{(X \cdot r)^{k_*}} \leq \sum_{k=0}^\infty \frac{t^k}{k!} \ex{}{(X \cdot r)^k}  = \ex{}{e^{t (X \cdot r)}} \leq e^{t^2 \norm{r}_2^2 /2}.$$
%
Rearranging and setting $t = \pm \sqrt{k_*}/\norm{r}_2$, we obtain $$ \ex{}{(X \cdot r)^{k_*}} \leq \frac{k_*!}{t^{k_*}} e^{t^2 \norm{r}_2^2 /2} = \frac{k_*! \norm{r}_2^{k_*} e^{k_*/2}}{\sqrt{k_*}^{k_*}} \leq \left( \frac{k_*^2 \norm{r}_2^2 e}{k_*} \right)^{k_*/2} = (e \norm{r}_2^2 k_*)^{k_*/2},$$ as required.
\end{proof}
Now we can prove the upper bound for $k$-wise independence using the connection between moment bounds and tail bounds \cite{SSS}.
\begin{proof}[Proof of Theorem \ref{thm:KwiseUpperBound}]
Note that, if $X \in \{\pm 1\}^n$ is $k$-wise independent, then $$\ex{}{(X \cdot r)^k} = \sum_{i_1 \cdots i_k \in [n]} \left( \prod_{j=1}^k r_{i_j} \right) \cdot \ex{}{\prod_{j=1}^k X_{i_j}}$$ is the same as for uniform $X$, as this is the expectation of a degree-$k$ polynomial.
By Lemma \ref{lem:MomentBound} and Markov's inequality, we have (assuming $k$ is even), $$\pr{}{|X \cdot r| \geq T} = \pr{}{(X \cdot r)^k \geq T^k} \leq \frac{\ex{}{(X \cdot r)^k}}{T^k} \leq \left( \frac{e \norm{r}_2^2 k}{T^2} \right)^{k/2}.$$
Substituting $k = 2 \lceil\eta  \log_e(1/\delta) \rceil$ and $T=  e^{(\eta+1)/2\eta} \sqrt{k} \norm{r}_2$, we have $$\pr{}{|X \cdot r| \geq  T} \leq  \left( \frac{e \norm{r}_2^2 k}{(e^{(\eta+1)/2\eta} \sqrt{k} \norm{r}_2)^2} \right)^{\lceil \eta  \log_e(1/\delta) \rceil} = e^{-\lceil \eta \log_e(1/\delta) \rceil / \eta} \leq \delta.$$
\end{proof}

\end{document}